\documentclass[10pt,superscriptaddress,aps,pra,onecolumn, showpacs,nofootinbib,longbibliography,notitlepage]{revtex4-2}

\usepackage{pgfplots}
\pgfplotsset{compat=1.18}
\usepackage[utf8]{inputenc}
\usepackage[T1]{fontenc}     
\usepackage[british]{babel}  
\usepackage[sc,osf]{mathpazo}
\usepackage[scaled=0.86]{berasans}  
\usepackage[colorlinks=true, allcolors=blue, urlcolor=blue]{hyperref}  
\usepackage{graphicx} 
\usepackage[babel]{microtype}  
\usepackage{amsmath,amssymb,amsthm,bm,amsfonts,mathrsfs,bbm} 

\usepackage{xspace}  
\usepackage{pgfplots}
\usepackage{xcolor,colortbl}
\usepackage{array}
\usepackage{bigstrut}\usepackage[normalem]{ulem}
\usepackage{mathtools}
\usepackage{physics}

\usepackage{caption}

\usepackage{subcaption}
\usepackage{tabularx}
\newcolumntype{C}{>{\centering\arraybackslash}X}
\usepackage{lipsum}
\usepackage[inkscapelatex=true]{svg}
\usepackage{hyperref}
\hypersetup{
	colorlinks=true,
	linkcolor=red,
	filecolor=blue,      
	urlcolor=blue,
	citecolor=purple,
}

\newcounter{subeq}

\newcommand{\Q}{\mathbb{Q}}

\newcommand{\be}{\begin{equation}}
	\newcommand{\ee}{\end{equation}}
\newcommand{\bea}{\begin{eqnarray}}
	\newcommand{\eea}{\end{eqnarray}}
\newcommand{\bes}{\begin{equation*}}
	\newcommand{\ees}{\end{equation*}}
\newcommand{\beas}{\begin{eqnarray*}}
	\newcommand{\eeas}{\end{eqnarray*}}
\newcommand{\proj}[1]{\ket{#1}\bra{#1}}

\renewcommand{\H}{\mathcal{H}}

\def\tr{\mathrm{tr}}

\def\Q{\mathcal{Q}}

\newcommand{\vecP}{\vec{P}}
\newtheorem{thm}{Theorem}
\newtheorem*{thm*}{Theorem}

\newtheorem*{lem*}{Lemma}

\newtheorem*{prop*}{Proposition}

\newtheorem*{lipschitzLem*}{Lemma \ref{lipschitz}}
\newtheorem*{lipschitzCubeLem*}{Lemma \ref{lipschitzCube}}
\newtheorem*{pgmNearlyOptimalThm*}{Theorem \ref{pgmNearlyOptimal}}

\begin{document}
\title{Hardy-type self-testing and exposedness of tripartite GHZ correlations}

\author{Smritikana Patra}
\thanks{These authors contributed equally to this work.}
\email{asksmritikanapatra@gmail.com, palsoumyajit2210@gmail.com}
\affiliation{Department of Physics,
Indian Institute of Technology Jodhpur, Rajasthan 342037, India.}

\author{Soumyajit Pal}
\thanks{These authors contributed equally to this work.}
\email{asksmritikanapatra@gmail.com, palsoumyajit2210@gmail.com}
\affiliation{Physics and Applied Mathematics Unit, Indian Statistical Institute,
203 B. T. Road, Kolkata 700108, India.}

\author{Ranendu Adhikary}
 \email{ronjumath@gmail.com}
\affiliation{Department of Physics and Center for Quantum Frontiers of Research and Technology (QFort),
National Cheng Kung University, Tainan 701, Taiwan}

\begin{abstract}
{Nonlocality can be witnessed either through Bell-inequality violations or through logical contradictions such as Hardy's paradox. In the bipartite two input two outcome scenario, these two routes have distinct geometric behavior: CHSH-maximal correlations are exposed points of the quantum set, whereas known Hardy-type self-testing correlations on the no-signaling boundary are non-exposed. Here we show that this bipartite intuition fails in the tripartite two input two outcome scenario. We study the tripartite instance of a multipartite Hardy-type paradox and prove that the correlation attaining the maximal Hardy success probability self-tests the Greenberger--Horne--Zeilinger state and the associated measurements. Although this correlation lies on the no-signaling boundary, we show that it is an extremal and exposed point of the quantum correlation set. Moreover, it coincides with the correlation attaining the maximal violation of the Mermin inequality. Thus, in the tripartite GHZ scenario, the logical-paradox and Bell-inequality routes to nonlocality select the same exposed quantum boundary point. We also establish a robust version of the self-test, showing that small deviations from the ideal Hardy constraints imply quantitative closeness to the target state and measurements. Our results reveal a qualitative geometric difference between bipartite and tripartite Hardy-type nonlocality and suggest a broader investigation of exposedness for multipartite Hardy correlations in the multiparty setting.}
\end{abstract}
\maketitle

\section{Introduction}

Nonlocality can be witnessed in different, but closely related, ways. The most familiar route is through the violation of a Bell inequality, where a linear functional of the observed conditional probabilities is bounded for all local hidden-variable models but can be exceeded by quantum correlations \cite{brunner2014bell,Pironio2014allchsh}. A second route is provided by logical contradictions of the Hardy type, where a suitable pattern of zero and nonzero probabilities rules out any local-realistic description without explicitly invoking a Bell inequality \cite{hardy1993nonlocality,Goldstein1994hardy}. These two manifestations already appear in the bipartite $((2,2,2))$ Bell scenario, where two parties each choose between two binary-outcome measurements. In this setting, the CHSH inequality is the paradigmatic inequality-based witness of nonlocality \cite{clauser1969proposed,CHSH1970}, while Hardy's paradox gives a canonical probability-based contradiction with local realism \cite{hardy1993nonlocality,Goldstein1994hardy,scarani2019bell}.

The distinction between these two forms of nonlocality becomes particularly transparent from a convex-geometric viewpoint. In a Bell scenario, the set of local correlations forms a convex polytope whose vertices correspond to deterministic local strategies and whose facets are described by tight Bell inequalities \cite{Pitowsky1991,Peres1999allbellineq,bernards2020generalizing}. Quantum theory gives rise to a larger convex set, which contains nonlocal correlations lying outside the local polytope but still inside the no-signaling set \cite{brunner2014bell,Pironio2014allchsh}. In the $((2,2,2))$ scenario, quantum mechanics violates the CHSH inequality up to the Tsirelson bound \cite{Cirelson1980,Filipp2004generalizingtsirelson}. Hardy's paradox, on the other hand, detects nonlocality through a logical structure: if three specific joint probabilities vanish, then the presence of one strictly positive probability is incompatible with locality \cite{hardy1993nonlocality,Goldstein1994hardy,scarani2019bell}. The relation between these two forms of nonlocality has been studied from several perspectives \cite{Braun2008hardyvschsh,Mancinska2014unifiedhardychsh,goh2018geometry,Seiler2022hardychsh,Le2023quantumcorrelations,Chen2023quantumcorrelations,Barizien2025}.

A useful geometric distinction is that between exposed and non-exposed points of the quantum correlation set. An exposed point is the unique maximizer of some linear functional, or equivalently a point singled out by a supporting hyperplane. An extremal point that is not uniquely selected in this way is non-exposed \cite{Boyd2009book}. Bell-inequality-based self-tests, such as the self-test obtained from maximal CHSH violation, naturally lead to exposed quantum boundary points \cite{Thinh2019geometry,Barizien2024tsirelson}. Hardy-type correlations behave differently in the bipartite case: the quantum correlation attaining the maximal Hardy success probability is known to be an extremal but non-exposed point of the quantum set \cite{goh2018geometry,Chen2023quantumcorrelations}. More generally, several bipartite Hardy-type self-testing correlations found on the no-signaling boundary are provably non-exposed \cite{Chen2023quantumcorrelations}, and further criteria identify broad families of bipartite quantum realizations whose corresponding points are non-exposed \cite{Barizien2025}. Thus, in the $((2,2,2))$ scenario, inequality-based and paradox-based nonlocality lead to sharply different geometric behavior.

This contrast raises a natural question: is non-exposedness a generic geometric feature of Hardy-type logical nonlocality, or is it specific to the bipartite setting? Multipartite Bell scenarios provide a natural arena in which to address this question. The $((3,2,2))$ case is especially important because it contains two canonical routes to GHZ-type nonlocality. On the inequality side, the Mermin inequality provides a fundamental witness of tripartite nonlocality \cite{mermin1990extreme}. On the logical-paradox side, multipartite Hardy-type arguments extend Hardy's original contradiction to $N$ parties \cite{Wu1996hardy,Ghosh1998hardy,choudhary2011hardygeneralized,Das2013tripartitehardy,cereceda2004hardy}. Their respective generalizations, including the Mermin-Ardehali-Belinskii-Klyshko inequalities \cite{Mermin1990,Ardehali1992,belinskiui1993interference} and multipartite Hardy-type paradoxes \cite{cereceda2004hardy}, are closely connected to Greenberger--Horne--Zeilinger (GHZ) correlations \cite{Chen2009tightmultipartitebell}. The tripartite GHZ scenario is therefore an ideal setting in which to compare the geometric consequences of inequality-based and paradox-based nonlocality.

This geometric perspective is further enriched by self-testing, a device-independent certification method that identifies the underlying quantum state and measurements solely from the observed correlations \cite{MY2004,kaniewski2016analytic,kaniewski2017self,vsupic2020self,christandl2022operational,baptista2025mathematical}. Self-testing is closely related to the boundary structure of the quantum correlation set: self-tested correlations are typically extremal, and in many Bell-inequality-based examples they are exposed by the Bell functional whose maximal quantum value they attain \cite{vsupic2020self,goh2018geometry}. For example, maximal CHSH violation self-tests the two-qubit maximally entangled state and appropriate local measurements \cite{Wang2016,valcarce2022self}, while maximal Mermin violation self-tests the tripartite GHZ state \cite{chen2004maximalmermin,kaniewski2016analytic,kaniewski2017self,panwar2023elegant,bao2025Graphtheoretic}. More generally, GHZ states admit self-testing from a variety of inequivalent Bell inequalities \cite{Supic2018,sarkar2022self,Singh2025ghz}, reflecting the rich structure of multipartite quantum correlations.

In this work, we focus on the tripartite instance of a multipartite Hardy-type paradox \cite{cereceda2004hardy} in the $((3,2,2))$ Bell scenario. We show that the quantum correlation attaining the maximal Hardy success probability self-tests the tripartite GHZ state and the associated local measurements. This provides a Hardy-type, rather than Bell-inequality-based, certification of the GHZ realization. More importantly, the resulting correlation displays a geometric behavior that is qualitatively different from the known bipartite Hardy-type examples. Although the ideal Hardy behavior lies on the boundary of the no-signaling set by saturating several positivity constraints, we prove that the tripartite Hardy-maximal correlation is not only extremal in the quantum set but also exposed. This contrasts with the bipartite case, where Hardy-type self-testing correlations on the no-signaling boundary are known to be non-exposed \cite{Chen2023quantumcorrelations}.

A further consequence of our analysis is that this same Hardy-maximal correlation coincides with the correlation attaining the maximal quantum violation of the Mermin inequality. A priori, the two criteria select correlations in very different ways: the Mermin inequality optimizes a linear Bell functional, whereas the Hardy construction imposes a logical pattern of vanishing and nonvanishing probabilities. Our result shows that, in the tripartite GHZ scenario, these two routes converge to the same quantum boundary point. Thus, the maximal tripartite Hardy correlation is simultaneously a Hardy-maximal behavior, a Mermin-maximal behavior, a self-testing correlation for the GHZ state, and an exposed extremal point of the quantum correlation set.

Our results reveal a qualitative geometric transition from the bipartite to the tripartite Hardy setting. In the $((2,2,2))$ scenario, known Hardy-type self-testing correlations on the no-signaling boundary provide prominent examples of non-exposed quantum points. In the $((3,2,2))$ scenario, however, the maximal Hardy-type correlation also lies on the no-signaling boundary but is nevertheless exposed. Moreover, it coincides with the Mermin-maximal correlation. This Hardy--Mermin coincidence suggests that multipartite Hardy-type nonlocality has a richer geometric structure than its bipartite counterpart and motivates the investigation of exposedness for maximal Hardy correlations in the general $((N,2,2))$ scenario.

We also provide a robust self-testing analysis for the tripartite Hardy construction. Robustness is essential for experimental and device-independent applications, where observed correlations cannot satisfy ideal probability constraints exactly. We show that small deviations from the ideal Hardy constraints, together with a Hardy success probability close to its optimal quantum value, imply quantitative closeness of the underlying state and measurements to the ideal GHZ realization. This places the Hardy-type construction on an operational footing comparable to robust self-tests based on Bell inequalities \cite{kaniewski2016analytic,kaniewski2017self,Supic2018,sarkar2022self,Singh2025ghz}. Beyond its foundational significance, self-testing plays a central role in device-independent quantum information protocols, including quantum key distribution, randomness certification, network certification, and delegated quantum computation \cite{Metger2021selftest,Sekatski2021deviceindependent,RoyDeloison2025DIQKD,Liu2024almostquantumbell,adhikary2025,Renou2018network,Supic2020certification,Bancal2021selftestingfinite,Murta2023network,Sarkar2025network,McKague2016BQP,Storz2025superconductingqubits,Adamson2025remotestateprep}.

The remainder of this paper is organized as follows. Section~\ref{sec1} introduces multipartite nonlocal correlations and the Hardy-type paradox used throughout the work. Section~\ref{sec2} develops the self-testing framework for the GHZ state based on this paradox. In Sec.~\ref{sec3}, we analyze the geometry of the quantum correlation attaining the maximal tripartite Hardy success probability, proving that it is an exposed extremal point of the quantum set and identifying it with the Mermin-maximal correlation. Section~\ref{sec4} presents the robust self-testing analysis in the tripartite scenario. Finally, Sec.~\ref{sec5} summarizes our results and discusses possible extensions to higher-party $((N,2,2))$ scenarios.

\section{Hardy's nonlocality argument}\label{sec1}

Consider $N$ spatially separated observers $A_{1}$, $A_{2}$, $\dots$, and $A_{N}$ engaged in rounds of measurements on $N$-partite physical systems. In each round, every party $A_{i}$ chooses a measurement setting $x_{i}$ and obtains an outcome $a_{i}$. The resulting joint conditional probability distribution is denoted as $P(a_1 a_2\dots a_N | x_1 x_2\dots x_N)$. A behavior $P(a_1 a_2\dots a_N | x_1 x_2\dots x_N)$ is said to admit a \emph{local} model if there exists a shared classical hidden variable $\lambda$ such that
\begin{equation}\label{eq:local_equation}
P_{\mathcal{L}}(a_1 a_2\dots a_N | x_1 x_2\dots x_N)
    = \sum_{\lambda} \nu(\lambda)\prod_{i=1}^{N}P_{A_{i}}(a_{i}| x_{i},\lambda),
\end{equation}
where $\nu(\lambda)\ge 0$ and $\sum_{\lambda}\nu(\lambda)=1$. A distribution that does not satisfy the decomposition~\eqref{eq:local_equation} is nonlocal~\cite{bancal2013definitions}.

We first recall Hardy's nonlocality argument in the bipartite scenario. Two observers $A_{1}$ and $A_{2}$ each choose between two dichotomic observables $A^{0}_{i}$ and $A^{1}_{i}$ ($i=1,2$), with outcomes in $\{+1,-1\}$. Let $P(a_{1},a_{2}\mid x_{1},x_{2})$ be the corresponding joint probability distribution. Hardy showed~\cite{hardy1992quantum} that if the four constraints,
\begin{equation}\label{hardy0}
\begin{split}
P(+1,+1|A^0_1,A^0_2)&=p>0,\\
P(+1,+1|A^1_1,A^0_2)&=0,\\
P(+1,+1|A^0_1,A^1_2)&=0,\\
P(-1,-1|A^1_1,A^1_2)&=0.
\end{split}
\end{equation}
are satisfied, then the behavior is necessarily nonlocal, i.e., it does not admit the form~\eqref{eq:local_equation}. The parameter $p=P(+1,+1| A^{0}_{1},A^{0}_{2})$ quantifies the strength of the paradox. Quantum mechanically, its maximal value is $\frac{-11 + 5\sqrt{5}}{2}$ $(\approx 0.09)$ \cite{RZS12} and is achieved with projective measurements on a pure two-qubit non maximally entangled state.

We now extend Hardy's construction to the multipartite setting. Consider $N$ observers sharing an $N$-partite system, where each observer $A_{i}$ chooses between two dichotomic observables $A^{0}_{i}$ and $A^{1}_{i}$ with outcomes $\{+1,-1\}$. Among the various possible generalizations \cite{Rahaman2014nonlocal,chen2014test}, we adopt the family introduced in Ref.~\cite{cereceda2004hardy}. The $N$-party Hardy-type conditions are

    \begin{eqnarray}\label{HardyN}
        P(+1,+1,\dots,+1|A^0_1,A^0_2,A^0_3,\dots,A^0_N)=p^H_N, \label{eq:HN} \\
            \begin{aligned} \label{multi}
			P(+1,+1,\dots,+1|A^1_1,A^0_2,A^0_3,\dots,A^0_N)&=0,\\
			P(+1,+1,\dots,+1|A^0_1,A^1_2,A^0_3,\dots,A^0_N)&=0,\\
                        \vdots&\\
                P(+1,+1,\dots,+1|A^0_1,A^0_2,A^0_3,\dots,A^1_N)&=0,\\
			P(-1,-1,\dots,-1|A^1_1,A^1_2,A^1_3,\dots,A^1_N)&=0.
		\end{aligned}  
            \end{eqnarray}

A behavior satisfying~\eqref{multi} is nonlocal if and only if $p^{H}_{N}>0$. A sketch of the proof is as follows. Suppose a behavior satisfies~\eqref{eq:local_equation} and the Hardy constraints with $p^{H}_{N}>0$. Then, $p^{H}_{N}=\sum_{\lambda}\nu(\lambda)\prod_{i=1}^{N}P_{A_{i}}(+1| A^{0}_{i},\lambda)>0$ implies the existence of at least one $\lambda_{0}$ such that $P_{A_{i}}(+1| A^{0}_{i},\lambda_{0})\neq 0$ for all $i$. The vanishing constraints in~\eqref{multi} then require $P_{A_{i}}(+1| A^{1}_{i},\lambda_{0})=0$ for each $i$, and hence $P_{A_{i}}(-1| A^{1}_{i},\lambda_{0})\neq 0$ for all $i$. This assigns a strictly positive probability to $P(-1,-1,\dots,-1|A^1_1,A^1_2,A^1_3,\dots,A^1_N)$, contradicting the last condition in~\eqref{multi}. Thus $p^{H}_{N}>0$ implies nonlocality.

In this paper, we focus on the tripartite instance $N=3$, namely the $((3,2,2))$ Bell scenario. We study the quantum behavior that maximizes the Hardy success probability $p_3^H$, and show that this behavior self-tests the GHZ state, and nevertheless defines an exposed extremal point of the quantum correlation set.

\section{Self-testing of GHZ state}\label{sec2}
We begin by recalling the notion of self-testing for quantum states and measurements, following the formulation in~\cite{vsupic2020self}. A quantum correlation $\vecP \in \Q$ is said to \emph{self-test} a reference quantum realization $\{\ket{\widetilde{\psi}}, \{\widetilde{M_{a_1|x_1}}\}, \{\widetilde{M_{a_2|x_2}}\}, \{\widetilde{M_{a_3|x_3}}\}\}$ if the observed correlation uniquely determines this realization up to local isometries.  More precisely, whenever $\vecP$ is produced by some state $\ket{\psi}_{A_1A_2A_3}$ and local measurements $\{M_{a_1|x_1}\}$, $\{M_{a_2|x_2}\}$, $\{M_{a_3|x_3}\}$ there must exist local isometries $\Phi_{A_i}: \H_{A_i}\mapsto\H_{A'_i}\otimes \H_{A''_i}$ such that
\begin{equation}\label{def:Selftest}
    \Phi (M_{a_1|x_1}\otimes M_{a_2|x_2}\otimes M_{a_3|x_3} \ket{\psi}_{A_1A_2 A_3}) =\ket{\varsigma}_{A'_1A'_2 A'_3}\otimes (\widetilde{M_{a_1|x_1}}\otimes\widetilde{M_{a_2|x_2}}\otimes\widetilde{M_{a_3|x_3}}\ket{\widetilde{\psi}}_{A''_1A''_2 A''_3}),
\end{equation}
where $\Phi:=\Phi_{A_1} \otimes \Phi_{A_2}\otimes \Phi_{A_3}$ and $\ket{\varsigma}_{A'_1A'_2A'_3}$ is some junk state. 

We now recall two characterizations of the extreme points of the quantum set $\Q$, relevant in this scenario. The first, due to~\cite{Mas06}, states that in an $N$-partite Bell experiment with two dichotomic measurements per party, every extreme point of $\Q$ can be realized using projective measurements on an $N$-qubit pure state. The second, due to~\cite{Barizien2025}, refines this further: all extreme points of $\Q$ in Bell scenario can be achieved using a \emph{real} $N$-qubit pure state together with \emph{real} unitary measurements.

With these facts in mind, it is natural to investigate whether seemingly different qubit strategies that reproduce a given correlation $\vecP$ are in fact physically equivalent. 


First, we show that the maximum achievable value of $p^H_3$ over a tripartite qubit system, $(\mathbb{C}^2)^{\otimes 3}$, is exactly $\frac{1}{8}$. Furthermore, as detailed in Appendix~\ref{appenA}, there is a unique qubit strategy (up to local unitaries) that reproduces this optimal correlation. Specifically, it requires the system to be in the GHZ state, with the observables for each party given by
\begin{equation}
\begin{aligned}
    A^0_i &= \begin{pmatrix}
    0 & e^{-\iota \frac{\pi}{6}} \\ 
    e^{\iota \frac{\pi}{6}} & 0
    \end{pmatrix} 
    & \text{and} \quad 
    A^1_i &= \begin{pmatrix}
    0 & e^{-\iota \frac{2\pi}{3}} \\ 
    e^{\iota \frac{2\pi}{3}} & 0
    \end{pmatrix}
\end{aligned}
\end{equation}

Subsequently, we demonstrate that the maximum achievable value of $p^H_3$ over a system of arbitrary local dimension, $(\mathbb{C}^d)^{\otimes 3}$, is identical to that of the qubit case. Consequently, we will see that the quantum state GHZ can be self-tested.

\begin{thm}\label{optimal3} 
 The maximum achievable value of $p^H_3$ among all three-qubit states represents the optimal value attainable within tri-partite quantum states of any finite dimension.
\end{thm}

\begin{proof}
Here we only present the outline of the proof as the details are quite similar to the proof given in \cite{RZS12,rai2021device,adhikary2024self,adhikary2024self1,Chen2023quantumcorrelations}.

Consider a general tripartite state $\rho$ shared among $A_1$, $A_2$, and $A_3$. The operator $\Pi_{a_i|x_i}$ represents the outcome $a_i$ obtained by $A_i$ when it measures the observable $x_i$. Thus, 
\begin{equation}
P(a_1,a_2,a_3|x_1,x_2,x_3) = \operatorname{Tr}[\rho (\Pi_{a_1|x_1} \otimes \Pi_{a_2|x_2} \otimes \Pi_{a_3|x_3})].    
\end{equation}

Considering no constraints on dimensionality, we employ Neumark's dilation theorem \cite{Neumark}, restricting our focus to projective measurements. Thus, the observables associated with $A_1$, $A_2$, and $A_3$ consist of Hermitian operators having eigenvalues of $\pm 1$, depicted as:
\begin{equation*}
x_i = (+1) \Pi_{+|x_i} + (-1) \Pi_{-|x_i}, \quad x_i \in \{A^0_i, A^1_i\}
\end{equation*}
with $\Pi_{+|x_i} + \Pi_{-|x_i}=\mathbb{I}$.

Now using the lemma stated in \cite{Mas06} there exists an orthonormal basis in which party $A_i$'s all four projectors $\Pi_{+|A^0_i}$, $\Pi_{-|A^0_i}$, $\Pi_{+|A^1_i}$ and $\Pi_{-|A^1_i}$ become simultaneously block diagonal, with each block of size either $1 \times 1$ or $2 \times 2$. Hence inducing a direct sum decomposition of the Hilbert space $\mathcal{H}_{A_i} = \bigoplus_k \mathcal{H}_{A_i}^k$ with every subspace $\mathcal{H}_{A_i}^k$ has dimension at most two. Each projector decomposes as $\Pi_{\pm|A^{0/1}_i}=\bigoplus_k\Pi_{\pm|A^{0/1}_i}^k$with $\Pi_{\pm|A^{0/1}_i}^k$ acting only within $\mathcal{H}_{A_i}^k$. The projector onto $\mathcal{H}_{A_i}^k$ is given by $\Pi^k=\Pi_{+|A^{0}_i}^k+\Pi_{-|A^{0}_i}^k=\Pi_{+|A^{1}_i}^k+\Pi_{-|A^{1}_i}^k$.

Hence, we get the following expression of the joint distribution:
\begin{equation}
P(a_1,a_2,a_3|x_1,x_2,x_3) = \sum_{l,m,n} \lambda_{lmn} \operatorname{Tr}[\rho_{lmn} (\Pi_{a_1|x_1}^l \otimes \Pi_{a_2|x_2}^m \otimes \Pi_{a_3|x_3}^n)] \equiv \sum_{l,m,n} \lambda_{lmn} P_{lmn}(a_1,a_2,a_3|x_1,x_2,x_3)
\end{equation}
where the coefficients $\lambda_{lmn}$ are determined as $\lambda_{lmn} = \operatorname{Tr}[\rho (\Pi_l \otimes \Pi_m \otimes \Pi_n)]$, with $\lambda_{lmn} \geq 0$, $\forall l,m,n$ and $\sum_{l,m,n} \lambda_{lmn} = 1$ with $\rho_{lmn} = \frac{\Pi^l \otimes \Pi^m \otimes \Pi^n \rho \Pi^l \otimes \Pi^m \otimes \Pi^n}{\lambda_{lmn}}$. $\rho_{lmn}$ is a
trace one positive operator acting on subspace of types $\mathbb{C}\otimes\mathbb{C}\otimes\mathbb{C}$,$\mathbb{C}^2\otimes\mathbb{C}\otimes\mathbb{C}$,$\dots$,$\mathbb{C}^2\otimes\mathbb{C}^2\otimes\mathbb{C}^2$.

If the joint probability $P(a_1,a_2,a_3|x_1,x_2,x_3)$ meets the conditions outlined in (\ref{multi}), then it implies that the joint probability $P_{lmn}(a_1,a_2,a_3|x_1,x_2,x_3)$ will also adhere to the constraint equations. Consequently, the maximum attainable value of $p^H_3$ across all three-qubit states denotes the optimal value achievable within tri-partite quantum states.
\end{proof}

\begin{thm}\label{self3}
If the maximum value of $p^H_3$ is observed, then the state of the system is equivalent up to local isometries
to $|\varsigma\rangle_{A'_1A'_2A'_3} \otimes |GHZ\rangle_{A''_1A''_2A''_3}$, where $|GHZ\rangle$ attains the maximum value of $p^H_3$  and $|\varsigma\rangle$ is an arbitrary tripartite state.
\end{thm}

\begin{proof}
 We can choose eigenstates of the observables $A^0_1$, $A^0_2$, and $A^0_3$ to be computational eigenbasis:

\begin{equation*}
\begin{split}
\Pi_{+|A^0_1}^l &= \vert 2l \rangle \langle 2l \vert, \quad \Pi_{-|A^0_1}^l = \vert 2l+1 \rangle \langle 2l+1 \vert, \\
\Pi_{+|A^0_2}^m &= \vert 2m \rangle \langle 2m \vert, \quad \Pi_{-|A^0_2}^m = \vert 2m+1 \rangle \langle 2m+1 \vert,\\
\Pi_{+|A^0_3}^n &= \vert 2n \rangle \langle 2n \vert, \quad \Pi_{-|A^0_3}^n = \vert 2n+1 \rangle \langle 2n+1 \vert
\end{split}
\end{equation*}
where $l$, $m$, and $n$ belong to the set $\{0, 1, 2, \dots\}$. Now, $p^H_{3,lmn}$, where $p^H_{3,lmn} = P_{lmn}(+1, +1, +1|A^0_1, A^0_2,A^0_3)$, in the subspace $H_{A_1}^l \otimes H_{A_2}^m \otimes H_{A_3}^n$ can attain the maximum value if and only if $\rho_{lmn} = |GHZ\rangle_{lmn}\langle GHZ|$. Therefore, the unknown state $|\chi\rangle$ can give the maximum value of $p^H_3$, if and only if,
\begin{equation*}
|\chi\rangle = \bigoplus_{l,m,n}\sqrt{\lambda_{lmn}}|GHZ\rangle_{lmn}.
\end{equation*}

Hence, if we choose the local isometries in the following way,
\begin{align*}
&\Phi_{A_i} = \Phi, \\
&\Phi |2t,0\rangle_{XX'} \rightarrow |2t,0\rangle_{XX'}, \\
&\Phi |2t+1,0\rangle_{XX'} \rightarrow |2t,1\rangle_{XX'}, 
\end{align*}
where $X'X'' \in \{A'_iA''_i\}$, then we have,
\begin{equation*}
 (\Phi_{A_1}\otimes\Phi_{A_2}\otimes\Phi_{A_3})|\chi\rangle_{A_1A_2A_3}= |\varsigma\rangle_{A'_1A'_2A'_3} \otimes |GHZ\rangle_{A''_1A''_2A''_3},   
\end{equation*}
where $|\varsigma\rangle_{A'_1A'_2A'_3}$ some junk state.
\end{proof}

\section{Geometric characterization of the Hardy-maximal correlation}\label{sec3}
We now examine the extremal quantum correlation $\vec{\mathcal{P}}_H$ associated with the tripartite Hardy paradox~\eqref{HardyN}. Given that an extremal point of the quantum correlation set may be either \emph{exposed} or \emph{non-exposed}, our aim is to determine the nature of extremality of this particular correlation. Recall that an extremal point $\mathcal{P}$ of a convex set $\mathcal{C}$ is called \emph{exposed} if there exists a supporting hyperplane whose intersection with $\mathcal{C}$ is exactly $\{\mathcal{P}\}$. If no such hyperplane singles out the point uniquely, then the extremal point is termed \emph{non-exposed}.  
Exposed points correspond to unique optimizers of linear functionals over the set, whereas non-exposed points arise as limits of exposed points but are not themselves uniquely determined by any linear functional.

In the tripartite scenario with two dichotomic measurements per party and binary outcomes, the set of no-signaling correlations forms a polytope of dimension $26$. A convenient way to represent any point of this polytope is through the full list of single-body, two-body, and three-body correlators. Accordingly, any correlation $\mathcal{P}\in\mathbb{R}^{26}$ can be expressed using the coordinate array
\begin{widetext}
\begin{equation}
\resizebox{0.7\columnwidth}{!}{$
\begin{array}{c||c|c|c|c|c|c|c|c|c}
& \langle B_0 \rangle & \langle B_1 \rangle & \langle C_0 \rangle & \langle C_1 \rangle &
\langle B_0 C_0 \rangle & \langle B_0 C_1 \rangle & \langle B_1 C_0 \rangle & \langle B_1 C_1 \rangle \\
\hline\hline
\langle A_0 \rangle &
\langle A_0 B_0 \rangle & \langle A_0 B_1 \rangle & \langle A_0 C_0 \rangle & \langle A_0 C_1 \rangle &
\langle A_0 B_0 C_0 \rangle & \langle A_0 B_0 C_1 \rangle & \langle A_0 B_1 C_0 \rangle & \langle A_0 B_1 C_1 \rangle \\
\langle A_1 \rangle &
\langle A_1 B_0 \rangle & \langle A_1 B_1 \rangle & \langle A_1 C_0 \rangle & \langle A_1 C_1 \rangle &
\langle A_1 B_0 C_0 \rangle & \langle A_1 B_0 C_1 \rangle & \langle A_1 B_1 C_0 \rangle & \langle A_1 B_1 C_1 \rangle \\
\end{array}
$}
\label{eq:tripartite-corr}
\end{equation}
\end{widetext}

Then $\vec{\mathcal{P}}_H$ has the following form
\begin{widetext}
\begin{equation}
\resizebox{0.9\columnwidth}{!}{$
\begin{array}{c||c|c|c|c|c|c|c|c|c}
& \langle B_0 \rangle=0 & \langle B_1 \rangle=0 & \langle C_0 \rangle=0 & \langle C_1 \rangle=0 &
\langle B_0 C_0 \rangle=0 & \langle B_0 C_1 \rangle=0 & \langle B_1 C_0 \rangle=0 & \langle B_1 C_1 \rangle=0 \\
\hline\hline
\langle A_0 \rangle=0 &
\langle A_0 B_0 \rangle=0 & \langle A_0 B_1 \rangle=0 & \langle A_0 C_0 \rangle=0 & \langle A_0 C_1 \rangle=0 &
\langle A_0 B_0 C_0 \rangle=0 & \langle A_0 B_0 C_1 \rangle = -1 & \langle A_0 B_1 C_0 \rangle = -1 & \langle A_0 B_1 C_1 \rangle=0 \\
\langle A_1 \rangle=0 &
\langle A_1 B_0 \rangle=0 & \langle A_1 B_1 \rangle=0 & \langle A_1 C_0 \rangle=0 & \langle A_1 C_1 \rangle=0 &
\langle A_1 B_0 C_0 \rangle = -1 & \langle A_1 B_0 C_1 \rangle=0 & \langle A_1 B_1 C_0 \rangle=0 & \langle A_1 B_1 C_1 \rangle = 1 \\
\end{array}
$}
\label{eq:hardy-corr}
\end{equation}
\end{widetext}

We will now show that this point is exposed in the quantum set by finding the Bell function $\vec{\mathcal{B}}$ that is uniquely maximized by $\mathcal{P}_H$. Define operators
\begin{equation*}
\begin{array}{lcl}
G_{1} = A_{0} \otimes \mathbb{1} \otimes \mathbb{1} & G_{7} = A_{0} \otimes B_{0} \otimes \mathbb{1} & G_{19} = A_{0} \otimes B_{0} \otimes C_{0} \\
G_{2} = A_{1} \otimes \mathbb{1} \otimes \mathbb{1} & G_{8} = A_{0} \otimes B_{1} \otimes \mathbb{1} & G_{20} = A_{0} \otimes B_{1} \otimes C_{0} \\
G_{3} = \mathbb{1} \otimes B_{0} \otimes \mathbb{1} & G_{9} = A_{1} \otimes B_{0} \otimes \mathbb{1} & G_{21} = A_{1} \otimes B_{0} \otimes C_{0} \\
G_{4} = \mathbb{1} \otimes B_{1} \otimes \mathbb{1} & G_{10} = A_{1} \otimes B_{1} \otimes \mathbb{1} & G_{22} = A_{1} \otimes B_{1} \otimes C_{0} \\
G_{5} = \mathbb{1} \otimes \mathbb{1} \otimes C_{0} & G_{11} = \mathbb{1} \otimes B_{0} \otimes C_{0} & G_{23} = A_{0} \otimes B_{0} \otimes C_{1} \\
G_{6} = \mathbb{1} \otimes \mathbb{1} \otimes C_{1} & G_{12} = \mathbb{1} \otimes B_{0} \otimes C_{1} & G_{24} = A_{0} \otimes B_{1} \otimes C_{1} \\
\relax & G_{13} = \mathbb{1} \otimes B_{1} \otimes C_{0} & G_{25} = A_{1} \otimes B_{0} \otimes C_{1} \\
\relax & G_{14} = \mathbb{1} \otimes B_{1} \otimes C_{1} & G_{26} = A_{1} \otimes B_{1} \otimes C_{1} \\
\relax & G_{15} = A_{0} \otimes \mathbb{1} \otimes C_{0} & \relax \\
\relax & G_{16} = A_{0} \otimes \mathbb{1} \otimes C_{1} & \relax \\
\relax & G_{17} = A_{1} \otimes \mathbb{1} \otimes C_{0} & \relax \\
\relax & G_{18} = A_{1} \otimes \mathbb{1} \otimes C_{1} & \relax \\
\end{array}
\end{equation*}

Let $\vec{\mathcal{B}}$ be an arbitrary Bell function $\{b_1,b_2,\ldots,b_{26}\}$ and the corresponding Bell operator equals
\begin{equation*}
W = \sum_{j = 1}^{26} b_{j} G_{j}.
\end{equation*}

The eigenvalue equation $W\ket{\psi_H} = \lambda \ket{\psi_H}$ implies that
\begin{eqnarray}
\bra{001} W \ket{\psi_H} &=& 0,\\
\bra{010} W \ket{\psi_H} &=& 0,\\
\bra{100} W \ket{\psi_H} &=& 0,\\
\bra{101} W \ket{\psi_H} &=& 0,\\
\bra{011} W \ket{\psi_H} &=& 0,\\
\bra{110} W \ket{\psi_H} &=& 0
\end{eqnarray}
This condition can be expressed as a linear constraint of the form $\vec{\mathcal{B}} \cdot \vec{\mathcal{T}}_{i} = 0$, where the components of $\vec{\mathcal{T}}_{i}$ are proportional to the quantities $\bra{i} G_j \ket{\psi_H}$. 

Our objective is to determine the largest value of the Bell expression $\vec{\mathcal{B}}\cdot \mathcal{P}_{H}$ among all Bell functionals that are maximised by $\vec{\mathcal{P}}_H$. To formulate this optimisation as a linear program, we impose a convenient normalisation, for example that every local deterministic behaviour evaluates to at most \(1\) under $\vec{\mathcal{B}}$. Since any Bell functional can be rescaled, this restriction involves no loss of generality. The corresponding linear program is given by
\begin{gather*}
\begin{array}{ll}
\max & \vec{\mathcal{B}}\cdot \mathcal{P}_{H} \\[1mm]
\textnormal{over} & \vec{\mathcal{B}}\in \mathbb{R}^{26} \\[1mm]
\textnormal{subject to} &
\vec{\mathcal{B}}\cdot \vec{\mathcal{T}}_{i} = 0,
\quad i\in\{001,010,100,101,011,110\} ,\\
&\vec{\mathcal{B}}\cdot \mathcal{P}_{j} \leq 1,
\quad j = 1,2,\ldots,64,
\end{array}
\end{gather*}
where $\mathcal{P}_{j}$ ranges over the $64$ deterministic local behaviours. The maximum value of the linear program is found to be identically 2. The Bell function returned by the program is $\vec{\mathcal{B}}:=$ $\{b_1=0,\dots,b_{19}=0,b_{20}=-\frac{1}{2},b_{21}=-\frac{1}{2},b_{22}=0,b_{23}=-\frac{1}{2},b_{24}=0,b_{25}=0,b_{26}=\frac{1}{2}\}$. So $W_{\text{max}}=\frac{1}{2}(-A_{0} \otimes B_{1} \otimes C_{0}-A_{1} \otimes B_{0} \otimes C_{0}-A_{0} \otimes B_{0} \otimes C_{1}+A_{1} \otimes B_{1} \otimes C_{1})$.

If a Bell functional $\vec{\mathcal{B}}$ achieves its maximum value on $\vec{\mathcal{P}}_H$, then the corresponding Hardy state $\ket{\psi_{H}}$ must be an eigenstate of the associated Bell operator $W$. This means that there exists a scalar $\lambda$ such that $W \ket{\psi_{H}} = \lambda \ket{\psi_{H}} $. This condition ensures that the supporting hyperplane defined by $\vec{\mathcal{B}}$ is tangent to the boundary of the quantum set precisely at $\vec{\mathcal{P}}_H$. In our case, we get $W_{\text{max}} \ket{GHZ} = 2\cdot \ket{GHZ} $. So $W_{\text{max}}$ is tangent to the boundary of the quantum set precisely at $\vec{\mathcal{P}}_H$, proving the exposed nature of $\vec{\mathcal{P}}_H$.

The optimality of $\vec{\mathcal{B}}$ can be shown analytically as follows. First, note that $\vec{\mathcal{B}}$ satisfies the constraints defining the linear program. To further show that $\vec{\mathcal{B}} \cdot \mathcal{P}_H=2$ is the optimal max-value, we write the dual program
\begin{gather*}
\begin{array}{ll}
\min & \sum_{k = 1}^{64} y_{k}\\
\textnormal{over} & y_{k} \geq 0, \; y, z_i \in \mathbb{R}\\
\textnormal{subject to} & \sum_{k = 1}^{64} y_{k} \mathcal{P}_{k} + \sum_{i} z_i\vec{\mathcal{T}}_{i}  = \mathcal{P}_H.
\end{array}
\end{gather*}
The assignment
\begin{align*}
y_{k} &=
\begin{cases}
\frac{1}{8} &\text{if } k \in \left\{
\begin{aligned}
&1,5,15,19,23,24,32,\\
&38,44,47,51,52,54,57
\end{aligned}
\right\},\\
0.0721688 &\text{if } k \in \left\{
\begin{aligned}
7,45
\end{aligned}
\right\},\\
0.0528312 &\text{if } k \in \left\{
\begin{aligned}
10,36
\end{aligned}
\right\},\\
0 &\text{otherwise},
\end{cases}
\end{align*}

\begin{align*}
z_{i} &=
\begin{cases}
-0.408248 &\text{if } k \in \left\{
\begin{aligned}
1,6
\end{aligned}
\right\},\\
0.408248 &\text{if } k \in \left\{
\begin{aligned}
2,5
\end{aligned}
\right\},\\
0 &\text{otherwise},
\end{cases}
\end{align*}

is a valid solution to the dual and the resulting value is $2$. This concludes the proof of the optimality of the Bell functional $\vec{\mathcal{B}}$, and establishes that it is equivalent up to relabeling of the Mermin inequality \cite{mermin1990extreme}.

\section{Robust self-testing of Hardy state and measurement}\label{sec4}

The tripartite Hardy paradox~\eqref{HardyN} features the zero-probability constraints~\eqref{multi}. Of course, demanding strictly vanishing probabilities is unrealistic in any physical experiment. It is therefore essential to assess how robust a given self-testing statement is against small imperfections. Let us recall the SWAP method \cite{yang2014robust,bancal2015physical,Chen2023quantumcorrelations} and explain how to use it to obtain numerical robustness bounds for the Hardy state and the associated measurements.

We consider local operators $\mathsf{S}_{A_iA'_i}$ acting jointly on a black-box system $A_i$ and a trusted auxiliary system $A'_i$, and define
\begin{equation}
    \mathcal{S}\!\left(\rho_{A_1A_2A_3}\otimes\proj{000}_{A'_1A'_2A'_3}\right)\mathcal{S}^\dagger,
\end{equation}
where $\mathcal{S}=\mathsf{S}_{A_1A'_1}\otimes\mathsf{S}_{A_2A'_2}\otimes\mathsf{S}_{A_3A'_3}$. In the ideal case where the actual observables $A_i^j$ coincide with the reference ones $\widetilde{A}_i^j$, we choose $\mathsf{S}_{A_iA'_i}$ to swap the Hilbert spaces $\mathcal{H}_{A_i}$ and $\mathcal{H}_{A'_i}$. The fidelity
\begin{equation}\label{eq:statefidelity}
    F=\bra{\widetilde{\psi}}\rho_{\text{\tiny SWAP}}\ket{\widetilde{\psi}},
\end{equation}
between the reference state $\ket{\widetilde{\psi}}$ and the \emph{swapped} state
\begin{equation}\label{eq:rho_swap}
    \rho_{\text{\tiny SWAP}}=
    \tr_{A_1A_2A_3}\!\left[
    \mathcal{S}\left(\rho_{A_1A_2A_3}\otimes\proj{000}_{A'_1A'_2A'_3}\right)\mathcal{S}^\dagger
    \right],
\end{equation}
then quantifies how close the unknown shared state $\rho_{A_1A_2A_3}$ is to $\ket{\widetilde{\psi}}$. Note that, the entries of $\rho_{\text{\tiny SWAP}}$ are given by linear combinations of correlation terms from the set $d=\{\tr(\mathbb{1}\rho_{A_1A_2A_3}),\tr(A^0_1\rho_{A_1A_2A_3}),\ldots, \tr(A^0_1A^1_2A^0_3\rho_{A_1A_2A_3}),\ldots\}$. The fidelity $F$ is hence a linear combination of these moments.

The tripartite Hardy paradox~\eqref{HardyN} has maximal violation $p^{\max}_3=\frac{1}{8}$ and the corresponding quantum state is $\ket{GHZ}$. The observables are taken as
\begin{align*}
A^0_1=A^0_2=A^0_3=\begin{pmatrix}
    0 & e^{-\iota \frac{\pi}{6}} \\ 
    e^{\iota \frac{\pi}{6}} & 0
    \end{pmatrix},\\
A^1_1=A^1_2=A^1_3=\begin{pmatrix}
    0 & e^{-\iota \frac{2\pi}{3}} \\ 
    e^{\iota \frac{2\pi}{3}} & 0
    \end{pmatrix},
\end{align*}

We will now construct the local isometry that will work as swap between the Hilbert space $A_i$ and auxiliary Hilbert space $A'_i$. We will follow the same procedure as depicted in \cite{yang2014robust,bancal2015physical,Chen2023quantumcorrelations}. For each $A_i$ we define
\begin{align}\label{eq:localiso}
    \mathsf{S}_{A_iA'_i}:=U_{A_iA'_i}V_{A_iA'_i}U_{A_iA'_i},
\end{align}
where
\begin{align}
   U_{A_iA'_i}&:=\mathbb{1}_{A_i}\otimes\proj{0}_{A'_i}+X_{A_i}\otimes\proj{1}_{A'_i},\nonumber\\
   V_{A_iA'_i}&:=\frac{\mathbb{1}_{A_i}+Z_{A_i}}{2}\otimes\mathbb{1}_{A'_i}
   +\frac{\mathbb{1}_{A_i}-Z_{A_i}}{2}\otimes\sigma_{x,{A'_i}}.\nonumber
\end{align}

Consider, $\rho_{A_1A_2A_3}$ to be a general tripartite qudit state. Replacing \eqref{eq:localiso} in \eqref{eq:rho_swap}, we get

\begin{equation}\label{eq:rho_swap3}
     \rho_{\text{\tiny SWAP}}=
    \tr_{A_1A_2A_3}\!\left[
    \mathcal{S}\left(\rho_{A_1A_2A_3}\otimes\proj{000}_{A'_1A_2A'_3}\right)\mathcal{S}^\dagger
    \right]\\
    =\sum_{\mathsf{ijklst}}\mathsf{C}_{\mathsf{ijklst}}\;
\ket{\mathsf{i}}\bra{\mathsf{l}}\otimes
\ket{\mathsf{j}}\bra{\mathsf{s}}\otimes
\ket{\mathsf{k}}\bra{\mathsf{t}},
\end{equation}

where
\[
\mathsf{C}_{\mathsf{ijklst}}=\frac{1}{64}\,
\tr\!\Big[\big(\mathsf{M}^{A_1}_{\mathsf{il}}\otimes \mathsf{M}^{A_2}_{\mathsf{js}}\otimes \mathsf{M}^{A_3}_{\mathsf{kt}}\big)\rho_{A_1A_2A_3}\Big],
\]
and
\begin{align*}
\mathsf{M}^{A_1}_{\mathsf{il}}&=(\mathbb{1}+Z_{A_1})^{1-\mathsf{l}}(X_{A_1}-Z_{A_1}X_{A_1})^{\mathsf{l}}
(\mathbb{1}+Z_{A_1})^{1-\mathsf{i}}(X_{A_1}-X_{A_1}Z_{A_1})^{\mathsf{i}},\\
\mathsf{M}^{A_2}_{\mathsf{js}}&=(\mathbb{1}+Z_{A_2})^{1-\mathsf{s}}(X_{A_2}-Z_{A_2}X_{A_2})^{\mathsf{s}}
(\mathbb{1}+Z_{A_2})^{1-\mathsf{j}}(X_{A_2}-X_{A_2}Z_{A_2})^{\mathsf{j}},\\
\mathsf{M}^{A_3}_{\mathsf{kt}}&=(\mathbb{1}+Z_{A_3})^{1-\mathsf{t}}(X_{A_3}-Z_{A_3}X_{A_3})^{\mathsf{t}}
(\mathbb{1}+Z_{A_3})^{1-\mathsf{k}}(X_{A_3}-X_{A_3}Z_{A_3})^{\mathsf{k}}.
\end{align*}
Substituting \eqref{eq:rho_swap3} into \eqref{eq:statefidelity} yields
\begin{widetext}
\begin{equation}\label{figmeritstate}
\begin{split}
F =\;&
\tr\Big[\Big(\frac{1}{8}\mathbb{1} - \frac{1}{16} Z_{A_1} Z_{A_2} X_{A_3} - \frac{1}{16} Z_{A_1} X_{A_2} Z_{A_3} - \frac{1}{16} X_{A_1} Z_{A_2} Z_{A_3} + \frac{1}{32} X_{A_1} X_{A_2} X_{A_3} - \frac{1}{32} Z_{A_1} X_{A_1} Z_{A_2} X_{A_2} + \frac{1}{32} Z_{A_1} X_{A_1} X_{A_2} Z_{A_2}\\ &- \frac{1}{32} Z_{A_1} X_{A_1} Z_{A_3} X_{A_3} + \frac{1}{32}
Z_{A_1} X_{A_1} X_{A_3} Z_{A_3} + \frac{1}{32} X_{A_1} Z_{A_1} Z_{A_2} X_{A_2} - \frac{1}{32} X_{A_1} Z_{A_1} X_{A_2} Z_{A_2} + \frac{1}{32} X_{A_1} Z_{A_1} Z_{A_3} X_{A_3}\\& -
\frac{1}{32} X_{A_1} Z_{A_1} X_{A_3} Z_{A_3} - \frac{1}{32} Z_{A_2} X_{A_2} Z_{A_3} X_{A_3} + \frac{1}{32} Z_{A_2} X_{A_2} X_{A_3} Z_{A_3} + \frac{1}{32} X_{A_2} Z_{A_2}
Z_{A_3} X_{A_3} - \frac{1}{32} X_{A_2} Z_{A_2} X_{A_3} Z_{A_3}\\& + \frac{1}{16} Z_{A_1} X_{A_1} Z_{A_1} Z_{A_2} Z_{A_3} - \frac{1}{64} Z_{A_1} X_{A_1} Z_{A_1} X_{A_2} X_{A_3} +
\frac{1}{16} Z_{A_1} Z_{A_2} X_{A_2} Z_{A_2} Z_{A_3} + \frac{1}{16} Z_{A_1} Z_{A_2} Z_{A_3} X_{A_3} Z_{A_3}\\& - \frac{1}{64} X_{A_1} Z_{A_2} X_{A_2} Z_{A_2} X_{A_3} -
\frac{1}{64} X_{A_1} X_{A_2} Z_{A_3} X_{A_3} Z_{A_3} + \frac{1}{64} Z_{A_1} X_{A_1} Z_{A_1} Z_{A_2} X_{A_2} Z_{A_2} X_{A_3} + \frac{1}{64} Z_{A_1} X_{A_1} Z_{A_1} X_{A_2}
Z_{A_3} X_{A_3} Z_{A_3}\\& + \frac{1}{64} X_{A_1} Z_{A_2} X_{A_2} Z_{A_2} Z_{A_3} X_{A_3} Z_{A_3} - \frac{1}{64} Z_{A_1} X_{A_1} Z_{A_1} Z_{A_2} X_{A_2} Z_{A_2} Z_{A_3} X_{A_3} Z_{A_3}\Big)\rho_{A_1A_2A_3}\Big].
\end{split}
\end{equation}
\end{widetext}

Now to assess the robustness of the self-testing statements relative to the reference state, we compute the worst-case fidelity by optimizing over all quantum realizations consistent with the observed value of the Hardy paradox.  In particular, we employ a relaxation method known as the NPA hierarchy \cite{Navascues2007,Navascues2008}, which provides an infinite sequence of outer approximations of the quantum set $
\mathcal{Q}_1 \supset \mathcal{Q}_2 \supset \dots \supset \mathcal{Q}_\ell \supset \dots$. Each level in this hierarchy is defined via SDP. It has been proven that these sets converge to the quantum set in the limit $ \ell \to \infty $, i.e., $\lim_{\ell\to\infty}\mathcal{Q}_\ell = \mathcal{Q}$ \cite{Navascues2007,Navascues2008}. 

Now we solve the following optimization problem
\begin{equation}\label{eq:robust_state}
\begin{split}
     &\mathcal{F}=\min_{d\in \mathcal{Q}_\ell} F\\
     \text{s.t.}\;&\tr\left(\frac{\mathbb{1}+Z_{A_1}}{2} \otimes \frac{\mathbb{1}+Z_{A_2}}{2}\otimes \frac{\mathbb{1}+Z_{A_3}}{2} \rho_{A_1A_2A_3}\right)=p^{\max}_3-\varepsilon_1,\\
      &\tr\left(\frac{\mathbb{1}+X_{A_1}}{2} \otimes \frac{\mathbb{1}+Z_{A_2}}{2}\otimes \frac{\mathbb{1}+Z_{A_3}}{2} \rho_{A_1A_2A_3}\right)\le\varepsilon_2,\\
      &\tr\left(\frac{\mathbb{1}+Z_{A_1}}{2}\otimes \frac{\mathbb{1}+X_{A_2}}{2}\otimes \frac{\mathbb{1}+Z_{A_3}}{2} \rho_{A_1A_2A_3}\right)\le\varepsilon_2,\\
      &\tr\left(\frac{\mathbb{1}+Z_{A_1}}{2} \otimes \frac{\mathbb{1}+Z_{A_2}}{2}\otimes \frac{\mathbb{1}+X_{A_3}}{2} \rho_{A_1A_2A_3}\right)\le\varepsilon_2,\\
      &\tr\left(\frac{\mathbb{1}-X_{A_1}}{2} \otimes \frac{\mathbb{1}-X_{A_2}}{2}\otimes \frac{\mathbb{1}-X_{A_3}}{2} \rho_{A_1A_2A_3}\right)\le\varepsilon_2,
\end{split}
\end{equation}
where $\varepsilon_1,\varepsilon_2\ge 0$ capture admissible deviations. The above optimization problem has been implemented in PYTHON using CVXPY, by setting the hierarchy level as $\ell=5$ \cite{gitmishra}. The resulting fidelity estimates are presented in Fig.~\ref{fig:fidelity_hardy}.

Now for the robust self-testing of measurements, we use the same figure of merit as considered in \cite{yang2014robust,bancal2015physical,Chen2023quantumcorrelations}

\begin{equation}\label{figmeritmeasure}
\mathsf{T}_{A_i}
=\frac{1}{2}[
P(0|A^0_i,\ket{0})+P(1|A^0_i,\ket{1})\\
+P(0|A^1_i,\ket{+})+P(1|A^1_i,\ket{-})
]-1,
\end{equation}
where
\begin{align}\label{Eq:P_A}
P(a|A^j_i,\ket{\phi})=
\tr\left\{
\frac{\mathbb{1}+(-1)^a A^j_i}{2}\otimes\mathbb{1}_{A'_i}\;
\big[\Phi_{A_iA'_i}\big(\tr_{\{A_1,A_2,A_3\}\setminus\{A_i\}}\rho_{A_1A_2A_3}\otimes\proj{\phi}_{A'_i}\big)\Phi_{A_iA'_i}^\dagger\big]
\right\}.
\end{align}
When the devices implement the reference measurements, each $\mathsf{T}_{A_i}$ equals $1$. As in our case the system $A_1$, $A_2$ and $A_3$ have same measurement, we only consider the analysis for system $A_1$. In our explicit case,
\begin{equation}
\begin{split}
P(0|A^0_1,\ket{0})
&=\tr\!\left[
\Big(\tfrac{\mathbb{1}+Z_{A_1}}{2}
+\tfrac{\mathbb{1}-Z_{A_1}}{2}\,X_{A_1}\,\tfrac{\mathbb{1}+Z_{A_1}}{2}\,X_{A_1}\,\tfrac{\mathbb{1}-Z_{A_1}}{2}\Big)\rho_{A_1A_2A_3}
\right],\\
P(1|A^0_1,\ket{1})
&=\tr\!\left[
\Big(X_{A_1}\,\tfrac{\mathbb{1}-Z_{A_1}}{2}\,X_{A_1}
+X_{A_1}\,\tfrac{\mathbb{1}+Z_{A_1}}{2}\,X_{A_1}\,\tfrac{\mathbb{1}-Z_{A_1}}{2}\,X_{A_1}\,\tfrac{\mathbb{1}+Z_{A_1}}{2}\,X_{A_1}\Big)\rho_{A_1A_2A_3}
\right],\\
P(0|A^1_1,\ket{+})
&=\tr\!\left[
\Big(\tfrac{\mathbb{1}+X_{A_1}}{2}
+\tfrac{\mathbb{1}-X_{A_1}}{2}\,Z_{A_1}\,\tfrac{\mathbb{1}+X_{A_1}}{2}\,Z_{A_1}\,\tfrac{\mathbb{1}-X_{A_1}}{2}\Big)\rho_{A_1A_2A_3}
\right],\\
P(1|A^1_1,\ket{-})
&=\tr\!\left[
\Big(Z_{A_1}\,\tfrac{\mathbb{1}+X_{A_1}}{2}\,Z_{A_1}
+Z_{A_1}\,\tfrac{\mathbb{1}-X_{A_1}}{2}\,Z_{A_1}\,\tfrac{\mathbb{1}+X_{A_1}}{2}\,Z_{A_1}\,\tfrac{\mathbb{1}-X_{A_1}}{2}\,Z_{A_1}\Big)\rho_{A_1A_2A_3}
\right].
\end{split}
\end{equation}
We then estimate the worst-case figures of merit by solving
\begin{equation}\label{eq:robust_meas}
\begin{split}
      &\tau_{A_1}=\min_{d\in \mathcal{Q}_6}\,\mathsf{T}_{A_1}\\
      \text{s.t.}\;&\tr\left(\frac{\mathbb{1}+Z_{A_1}}{2} \otimes \frac{\mathbb{1}+Z_{A_2}}{2}\otimes \frac{\mathbb{1}+Z_{A_3}}{2} \rho_{A_1A_2A_3}\right)=p^{\max}_3-\varepsilon_1,\\
      &\tr\left(\frac{\mathbb{1}+X_{A_1}}{2} \otimes \frac{\mathbb{1}+Z_{A_2}}{2}\otimes \frac{\mathbb{1}+Z_{A_3}}{2} \rho_{A_1A_2A_3}\right)\le\varepsilon_2,\\
      &\tr\left(\frac{\mathbb{1}+Z_{A_1}}{2}\otimes \frac{\mathbb{1}+X_{A_2}}{2}\otimes \frac{\mathbb{1}+Z_{A_3}}{2} \rho_{A_1A_2A_3}\right)\le\varepsilon_2,\\
      &\tr\left(\frac{\mathbb{1}+Z_{A_1}}{2} \otimes \frac{\mathbb{1}+Z_{A_2}}{2}\otimes \frac{\mathbb{1}+X_{A_3}}{2} \rho_{A_1A_2A_3}\right)\le\varepsilon_2,\\      &\tr\left(\frac{\mathbb{1}-X_{A_1}}{2} \otimes \frac{\mathbb{1}-X_{A_2}}{2}\otimes \frac{\mathbb{1}-X_{A_3}}{2} \rho_{A_1A_2A_3}\right)\le\varepsilon_2,
\end{split}
\end{equation}

Note that trivial measurements, $\frac{\mathbb{1}+(-1)^aA^j_1}{2}=\frac{\mathbb{1}_2}{2}$ for all $a,j$, already achieve $\mathsf{T}_{A_1}=0$. Hence, a nontrivial self-test requires $\tau_{A_1}>0$. The corresponding robustness analysis is shown in Fig.~\ref{fig:fidelity_measure}.
\begin{figure*}[htbp]
    \centering
    \begin{subfigure}[t]{0.48\linewidth}
        \centering
        \includegraphics[width=\linewidth]{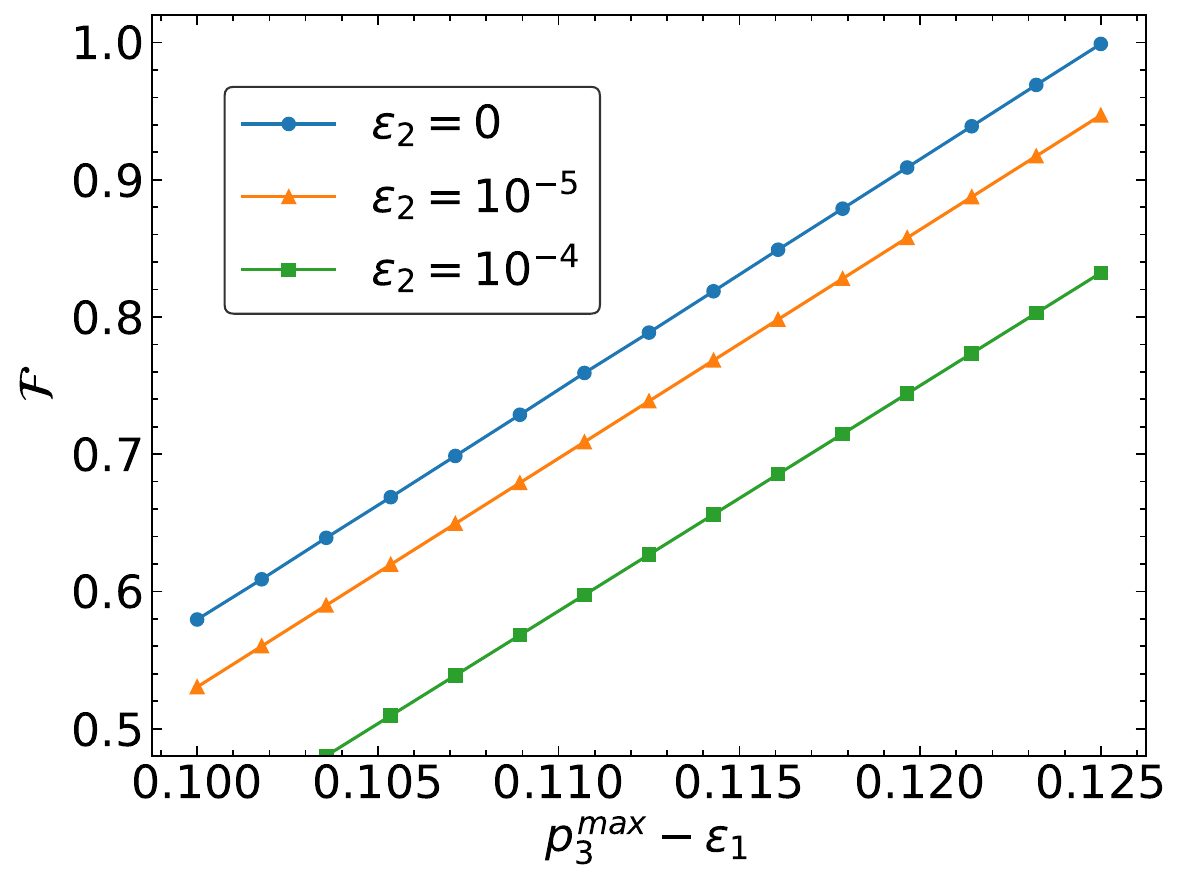}
        \caption{$\mathcal{F}$ corresponding to different values of $p^{max}_3 - \varepsilon_1$ and $\varepsilon_2$.}
        \label{fig:fidelity_hardy}
    \end{subfigure}\hfill
    \begin{subfigure}[t]{0.48\linewidth}
        \centering
        \includegraphics[width=\linewidth]{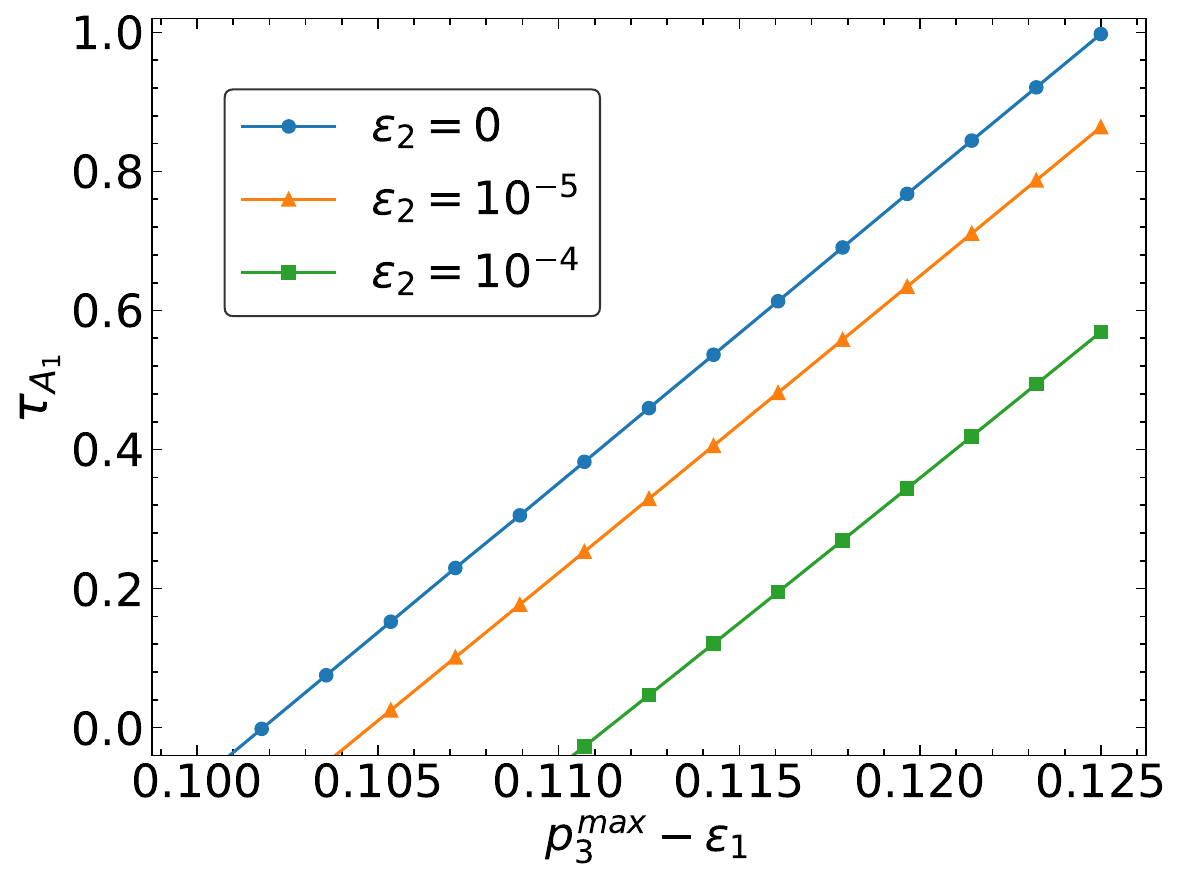}
        \caption{$\tau_{A_1}$ corresponding to different values of $p^{max}_3 - \varepsilon_1$ and $\varepsilon_2$.}
        \label{fig:fidelity_measure}
    \end{subfigure}
    \caption{  Plots illustrate the robustness of the self-testing result corresponding to the correlation that maximally violates the tripartite Hardy paradox. To quantify the quality of self-testing of the state, we adopt the fidelity with respect to the chosen reference state \eqref{figmeritstate} as the figure of merit. For the self-testing of the measurements, the relevant figure of merit is defined in \eqref{figmeritmeasure}. Because the reference measurements of $A_1$, $A_2$, and $A_3$ are similar, it is sufficient to display the plot for $A_1$ only. Throughout our analysis, the parameter $\varepsilon_2$ denotes the admissible deviation from the ideal zero-probability condition. All numerical results reported here are obtained using the level-6 outer approximation of the quantum set $\mathcal{Q}$.}
    \label{fig:robustness}
\end{figure*}

From Fig.~\ref{fig:robustness} we observe that, for $N=3$, the condition \eqref{multi} can tolerate at most an error of $\varepsilon_2=10^{-4}$. Within this regime, the fidelity satisfies $\mathcal{F}>\tfrac{1}{2}$ provided $\varepsilon_1 \lesssim 0.005$. For this range of $\varepsilon_1$, we also find that $\tau_{A_1}>0$. This robustness analysis guarantees that even when \eqref{multi} is only approximately satisfied, the state $\ket{\psi^H_N}$ and associated measurements can still be self-tested reliably.

\section{Conclusion}\label{sec5}

In this work, we studied the $N=3$ instance of a multipartite Hardy-type paradox in the $((3,2,2))$ Bell scenario. We showed that the quantum behavior attaining the maximal Hardy success probability self-tests the tripartite Greenberger--Horne--Zeilinger (GHZ) state and the associated measurements. The optimal Hardy success probability is found to be $p_3^H=1/8$, achieved by a realization locally equivalent to the GHZ state with suitable observables. This provides a Hardy-type, rather than Bell-inequality-based, route to device-independent certification of the GHZ realization.

A central message of our work is geometric. In the $((2,2,2))$ scenario, known Hardy-type self-testing correlations lying on the no-signaling boundary provide prominent examples of non-exposed quantum points \cite{goh2018geometry,Chen2023quantumcorrelations,Barizien2025}. The tripartite Hardy construction studied here behaves differently. Although the Hardy-maximal tripartite behavior also lies on the boundary of the no-signaling set, we prove that it is an extremal point of the quantum correlation set and, moreover, that it is exposed. Thus, Hardy-type logical nonlocality does not necessarily lead to non-exposed quantum boundary points. Already in the $((3,2,2))$ scenario, the geometry of the Hardy-maximal correlation differs qualitatively from the familiar bipartite picture.

We also showed that the Hardy-maximal tripartite behavior coincides with the behavior attaining the maximal quantum violation of the Mermin inequality. This establishes a direct connection between two a priori distinct ways of witnessing multipartite nonlocality. The Hardy construction selects a correlation through a logical pattern of vanishing and nonvanishing probabilities, whereas the Mermin inequality selects a correlation by optimizing a linear Bell functional. Our result shows that, in the tripartite GHZ scenario, these two routes converge to the same exposed quantum boundary point. In this sense, the same correlation admits simultaneously a logical-paradox interpretation, a Bell-inequality characterization, and a self-testing interpretation.

To address experimental imperfections, we further developed a robust self-testing analysis for the tripartite Hardy construction. We showed that if the Hardy zero-probability constraints are satisfied up to small deviations and the Hardy success probability remains close to its optimal quantum value, then the underlying state and measurements remain quantitatively close to the ideal GHZ realization, up to local isometries. In particular, our analysis demonstrates robustness for noise levels up to $10^{-4}$ in the zero-condition probabilities. This indicates that Hardy-type self-testing can be made meaningful beyond the idealized, exactly constrained setting.

Several open questions remain. The most immediate one is whether the same geometric phenomenon persists for the generalized $N$-party Hardy paradox with $N>3$. In particular, it would be interesting to determine whether the correlations attaining the maximal Hardy success probability in higher-party $((N,2,2))$ scenarios are also exposed extremal points of the quantum set, or whether new forms of non-exposed extremality appear. A related direction is to investigate whether the Hardy--Mermin coincidence observed here has an analogue for higher-party Mermin-Ardehali-Belinskii-Klyshko inequalities. More broadly, it would be valuable to understand whether other logical paradoxes, beyond the Hardy family, can lead to multipartite self-testing protocols with similarly sharp geometric characterizations. Such results would further clarify the relationship between paradox-based, inequality-based, and convex-geometric formulations of quantum nonlocality.

\section*{Acknowledgement.} We are grateful for stimulating discussions with Yeong-Cherng Liang,
Guruprasad Kar, Subhendu Bikash Ghosh, and Snehasish Roy Chowdhury. R.A. acknowledges financial support from the National Science and Technology Council, Taiwan (Grants No. 112-2628-M-006-007-MY4, 114-2811-M-006-069-MY2).

\bibliography{reference}
	

\appendix
\section{Detailed calculation of finding maximum $p^H_3$ over tripartite qubit system}\label{appenA}

Consider three parties $A_1$, $A_2$, and $A_3$ sharing a general pure three-qubit state of the form
\begin{equation}\label{eq:state}
\ket{\psi}_{A_1A_2A_3}
= \sum_{i,j,k \in \{0,1\}} a_{ijk} \ket{ijk},
\end{equation}
where the coefficients $a_{ijk} \in \mathbb{R}$ satisfy the normalization condition $\sum_{i,j,k} a_{ijk}^2 = 1$.

Without loss of generality, we fix one measurement for each party as $A_i^0 = \sigma_z$. The second measurement $A_i^1$ is taken to be
\begin{equation}\label{eq:measurement}
\begin{aligned}
\ket{A_i^{1,+}} &= \cos\!\left(\frac{\alpha_i}{2}\right)\ket{0} + \sin\!\left(\frac{\alpha_i}{2}\right)\ket{1},\\
\ket{A_i^{1,-}} &= \sin\!\left(\frac{\alpha_i}{2}\right)\ket{0} - \cos\!\left(\frac{\alpha_i}{2}\right)\ket{1},
\end{aligned}
\end{equation}
with $0 < \alpha_i < \pi$.

Imposing the zero-probability conditions~\eqref{multi}, we obtain the constraints
\begin{subequations}\label{eq:constraints}
\begin{align}
a_{100} &= -\cot\!\left(\frac{\alpha_1}{2}\right) a_{000}, \\
a_{010} &= -\cot\!\left(\frac{\alpha_2}{2}\right) a_{000}, \\
a_{001} &= -\cot\!\left(\frac{\alpha_3}{2}\right) a_{000}, \\
a_{111} &= a_{000} \tan\!\left(\frac{\alpha_1}{2}\right)\tan\!\left(\frac{\alpha_2}{2}\right)\tan\!\left(\frac{\alpha_3}{2}\right) \nonumber\\
&\quad - a_{100}\tan\!\left(\frac{\alpha_2}{2}\right)\tan\!\left(\frac{\alpha_3}{2}\right)
- a_{010}\tan\!\left(\frac{\alpha_1}{2}\right)\tan\!\left(\frac{\alpha_3}{2}\right) \nonumber\\
&\quad - a_{001}\tan\!\left(\frac{\alpha_1}{2}\right)\tan\!\left(\frac{\alpha_2}{2}\right)
+ a_{110}\tan\!\left(\frac{\alpha_3}{2}\right)
+ a_{101}\tan\!\left(\frac{\alpha_2}{2}\right)
+ a_{011}\tan\!\left(\frac{\alpha_1}{2}\right).
\end{align}
\end{subequations}

Substituting the first three relations into the last equation, we obtain
\begin{equation}\label{eq:linear_constraint}
a_{111} - a_{110} T_3 - a_{101} T_2 - a_{011} T_1 = K,
\end{equation}
where we define $T_i := \tan(\alpha_i/2)$ and
\begin{equation}
K = a_{000} \left( T_1 T_2 T_3 + \frac{T_2 T_3}{T_1} + \frac{T_1 T_3}{T_2} + \frac{T_1 T_2}{T_3} \right).
\end{equation}

Using normalization, maximizing $a_{000}^2$ is equivalent to minimizing $\sum_{(i,j,k)\neq(0,0,0)} a_{ijk}^2$. For fixed $a_{000}$, the coefficients $a_{100}$, $a_{010}$, and $a_{001}$ are determined, leaving four free variables. We thus consider the Lagrangian
\begin{equation}
\mathcal{L}
= a_{110}^2 + a_{101}^2 + a_{011}^2 + a_{111}^2
- \lambda \big( a_{111} - a_{110}T_3 - a_{101}T_2 - a_{011}T_1 - K \big).
\end{equation}

Stationarity conditions yield
\begin{equation}\label{eq:opt_coeff}
a_{110} = -\frac{\lambda}{2} T_3, \quad
a_{101} = -\frac{\lambda}{2} T_2, \quad
a_{011} = -\frac{\lambda}{2} T_1, \quad
a_{111} = \frac{\lambda}{2}.
\end{equation}

Substituting into Eq.~\eqref{eq:linear_constraint}, we obtain
\begin{equation}
\frac{\lambda}{2} = \frac{K}{1 + T_1^2 + T_2^2 + T_3^2},
\end{equation}
which gives
\begin{equation}\label{eq:opt_values}
\begin{aligned}
a_{110} &= -\frac{K T_3}{1 + T_1^2 + T_2^2 + T_3^2}, \\
a_{101} &= -\frac{K T_2}{1 + T_1^2 + T_2^2 + T_3^2}, \\
a_{011} &= -\frac{K T_1}{1 + T_1^2 + T_2^2 + T_3^2}, \\
a_{111} &= \frac{K}{1 + T_1^2 + T_2^2 + T_3^2}.
\end{aligned}
\end{equation}

Substituting back into the normalization condition yields
\begin{equation}\label{eq:a000_general}
a_{000}^2 =
\frac{1}{1 + \frac{1}{T_1^2} + \frac{1}{T_2^2} + \frac{1}{T_3^2}
+ \frac{\left(T_1 T_2 T_3 + \frac{T_2 T_3}{T_1} + \frac{T_1 T_3}{T_2} + \frac{T_1 T_2}{T_3} \right)^2}{1 + T_1^2 + T_2^2 + T_3^2}}.
\end{equation}

Thus maximizing $a_{000}^2$ is equivalent to minimizing the denominator of Eq.~\eqref{eq:a000_general} over all $T_1,T_2,T_3>0$. To perform this optimization, define $x_i=T_i^2, i=1,2,3$.

Let $D(x_1,x_2,x_3)$ denote the denominator in Eq.~\eqref{eq:a000_general}. Then
  \begin{equation}\label{eq:D_first}
  \begin{split}
    D(x_1,x_2,x_3)
  &=
  1+\frac{1}{x_1}+\frac{1}{x_2}+\frac{1}{x_3}
+\frac{
x_1x_2x_3
\left(
1+\frac{1}{x_1}+\frac{1}{x_2}+\frac{1}{x_3}
\right)^2}{1+x_1+x_2+x_3}\\&=\left( \frac{x_1 x_2 x_3 + x_1 x_2 + x_2 x_3 + x_1 x_3}{x_1 x_2 x_3} \right) \frac{(1+x_1)(1+x_2)(1+x_3)}{1 + x_1 + x_2 + x_3}.  
  \end{split}
\end{equation}

We now prove directly that $D(x_1,x_2,x_3)\geq 8$ for all $x_i>0$. Introduce the positive variables
\begin{equation}
u=\frac{1}{x_1},\qquad v=\frac{1}{x_2},\qquad w=\frac{1}{x_3}.
\end{equation}
Then Eq.~\eqref{eq:D_first} is equivalently
\begin{equation}\label{eq:D_uvw}
D = 1+u+v+w + \frac{(1+u+v+w)^2}{uv+uw+vw+uvw}.
\end{equation}

Let $t=u+v+w>0$. By the AM-GM inequality, we have
\begin{equation}
uv+uw+vw\leq \frac{t^2}{3}, \qquad uvw\leq \frac{t^3}{27}.
\end{equation}
Therefore,
\begin{equation}
uv+uw+vw+uvw \leq \frac{t^2}{3}+\frac{t^3}{27} = \frac{t^2(t+9)}{27}.
\end{equation}

Using this in Eq.~\eqref{eq:D_uvw}, we obtain the lower bound
\begin{equation}
D \geq 1+t+ \frac{27(1+t)^2}{t^2(t+9)}.
\end{equation}

A direct simplification gives
\begin{equation}
1+t+ \frac{27(1+t)^2}{t^2(t+9)} - 8 = \frac{(t-3)^2(t^2+8t+3)}{t^2(t+9)} \geq 0.
\end{equation}
Hence, $D(x_1,x_2,x_3)\geq 8$ for all $x_1,x_2,x_3>0$. Equality holds only when both AM-GM inequalities are saturated and $t=3$. This requires $u=v=w=1$, and therefore $x_1=x_2=x_3=1$.
Equivalently,
\begin{equation}
T_1=T_2=T_3=1, \qquad \alpha_1=\alpha_2=\alpha_3=\frac{\pi}{2}.
\end{equation}

Substituting $x_1=x_2=x_3=1$ into $D$ yields the minimum value 8. Thus, the absolute maximum is rigorously proven to be
\begin{equation}
a_{000}^2=\frac{1}{8}.
\end{equation}


Therefore, the optimal observables are $A_i^0 = \sigma_z$ and $A_i^1 = \sigma_x$, and the corresponding state is
\begin{equation}\label{eq:target_state}
\ket{\psi^*}_{A_1A_2A_3}
= \frac{1}{2\sqrt{2}} \big(
\ket{000} - \ket{001} - \ket{010} - \ket{100}
- \ket{011} - \ket{101} - \ket{110} + \ket{111}
\big).
\end{equation}

The state and measurements also equivalent to, GHZ state with the following observables,
\begin{equation}
\begin{aligned}
    A^0_i &= \begin{pmatrix}
    0 & e^{-\iota \frac{\pi}{6}} \\ 
    e^{\iota \frac{\pi}{6}} & 0
    \end{pmatrix} 
    & \text{and} \quad 
    A^1_i &= \begin{pmatrix}
    0 & e^{-\iota \frac{2\pi}{3}} \\ 
    e^{\iota \frac{2\pi}{3}} & 0
    \end{pmatrix}
\end{aligned}
\end{equation}
using the following local unitary
$\left(
\begin{array}{cc}
 -\frac{e^{\frac{11 i \pi }{12}}}{\sqrt{2}} & -\frac{e^{\frac{5 i
   \pi }{12}}}{\sqrt{2}} \\
 -\frac{e^{-\frac{11 i \pi }{12}}}{\sqrt{2}} & -\frac{e^{-\frac{5
   i \pi }{12}}}{\sqrt{2}} \\
\end{array}
\right)$.

\end{document}